\documentclass[12pt]{amsart}
\usepackage{enumerate}
\usepackage{hyperref}
\usepackage{graphicx,color}
\usepackage{cite}
\usepackage{amsthm, amsmath, amssymb, mathtools, braket, url}
\usepackage[top=45truemm, bottom=45truemm, left=30truemm, right=30truemm]{geometry}

\renewcommand{\epsilon}{\varepsilon}
\newcommand{\abs}[1]{\left|#1\right|}
\DeclareMathOperator{\Tr}{Tr}
\newcommand{\ketbra}[2]{\ket{#1}\!\bra{#2}}
\DeclareMathOperator{\diag}{diag}
\DeclareMathOperator*{\avg}{avg}
\DeclareMathOperator{\eig}{eig}
\DeclareMathOperator{\id}{id}
\newcommand{\Herm}{\mathsf{Herm}}
\newcommand{\PSD}{\mathsf{PSD}}

\numberwithin{equation}{section}

\newtheorem{theorem}{Theorem}[section]
\newtheorem{lemma}[theorem]{Lemma}
\newtheorem{corollary}[theorem]{Corollary}
\newtheorem{definition}[theorem]{Definition}
\newtheorem{proposition}[theorem]{Proposition}

\theoremstyle{definition}

\newtheorem*{example*}{Example}

\newcommand{\cE}{\mathcal{E}}

\newcommand{\cS}{\mathcal{S}}

\newcommand{\CQ}{\mathrm{CQ}}
\newcommand{\C}{\mathrm{C}}
\newcommand{\EC}{\mathrm{EC}}

\title[Mathematical comparison of classical and quantum mechanisms]{Mathematical comparison of classical and quantum mechanisms in optimization under local differential privacy}

\author[Y.\ Yoshida]{Yuuya Yoshida}
\address{Yuuya Yoshida\\
Graduate School of Mathematics\\ Nagoya University\\ Furo-cho\\ Chikusa-ku\\ Nagoya\\ 464-8602\\ Japan}
\curraddr{}
\email{m17043e@math.nagoya-u.ac.jp}



\begin{document}
\maketitle

\begin{abstract}
Let $\varepsilon>0$. An $n$-tuple $(p_i)_{i=1}^n$ of probability vectors is called $\varepsilon$-differentially private ($\varepsilon$-DP) if $e^\varepsilon p_j-p_i$ has no negative entries for all $i,j=1,\ldots,n$. An $n$-tuple $(\rho_i)_{i=1}^n$ of density matrices is called classical-quantum $\varepsilon$-differentially private (CQ $\varepsilon$-DP) if $e^\varepsilon\rho_j-\rho_i$ is positive semi-definite for all $i,j=1,\ldots,n$. Denote by $\mathrm{C}_n(\varepsilon)$ the set of all $\varepsilon$-DP $n$-tuples, and by $\mathrm{CQ}_n(\varepsilon)$ the set of all CQ $\varepsilon$-DP $n$-tuples. By considering optimization problems under local differential privacy, we define the subset $\mathrm{EC}_n(\varepsilon)$ of $\mathrm{CQ}_n(\varepsilon)$ that is essentially classical. Roughly speaking, an element in $\mathrm{EC}_n(\varepsilon)$ is the image of $(p_i)_{i=1}^n\in\mathrm{C}_n(\varepsilon)$ by a completely positive and trace-preserving linear map (CPTP map). In a preceding study, it is known that $\mathrm{EC}_2(\varepsilon)=\mathrm{CQ}_2(\varepsilon)$. In this paper, we show that $\mathrm{EC}_n(\varepsilon)\not=\mathrm{CQ}_n(\varepsilon)$ for every $n\ge3$, and estimate the difference between $\mathrm{EC}_n(\varepsilon)$ and $\mathrm{CQ}_n(\varepsilon)$ in a certain manner.
\end{abstract}

\section{Introduction}\label{intro}
In data analysis, 
data analysts need to know only some statistical information about private data while protecting the private data.
They hope to maximally utilize private data under some privacy protection.
In general, protection and utilization of private data have a trade-off relation, 
and researchers optimize the trade-off relation \cite{Kairouz,Holohan1,ISIT2019, Geng1,Geng2,Geng3}.

As a way protecting private data, 
Warner \cite{Warner} proposed \textit{randomized response} in 1965, 
in which private data $X$ are converted to other data $Y$ subject to a conditional probability distribution $\mathbb{P}_{Y|X}$, 
and the data $Y$ is released instead of $X$.
Since a data analyst collects only randomized data $Y$, 
private data $X$ are protected.

\vskip1ex
\textbf{Differential privacy.}---%
However, private data are not always protected in the above way.
For instance, if $Y$ is always equal to $X$, 
then it is clear that private data are not protected.
To enforce protection of data, 
we impose the following condition on the conditional probability distribution $\mathbb{P}_{Y|X}$: 
\begin{equation}
	\forall x,x',\ \mathbb{P}_{Y|X}(\cdot|x') \le e^\epsilon\mathbb{P}_{Y|X}(\cdot|x),
	\label{C-DP}
\end{equation}
where $\epsilon>0$ is a constant, and the inequality is entrywise.
This condition is called \textit{$\epsilon$-differential privacy ($\epsilon$-DP)} \cite{Dwork1,Dwork2,Duchi}.
Differential privacy (DP) was introduced by Dwork et al.\ \cite{Dwork1} and Dwork \cite{Dwork2} 
in the \textit{global privacy} context (the case when a company or government releases users' data partially for machine learning).
After that, DP was also introduced by Duchi et al.\ \cite{Duchi} 
in the \textit{local privacy} context (the case when data providers do not trust a data analyst).
Definition~\eqref{C-DP} is that in the local privacy context.

DP has been studied intensively by using classical probability theory.
However, there are only a few studies of quantum versions of DP \cite{qDP1,qDP2,qDP3,qDP4,ISIT2020} to the best of our knowledge.
In this paper, we define a quantum version of DP and investigate its mathematical aspects 
when $n$-ary data $X$ are converted to quantum states $\rho$ depending on $X$ (i.e., classical-quantum setting) in the local privacy context.

To define a quantum version of DP, we consider an $n$-tuple of quantum states $(\rho_x)_{x=1}^n$, 
where $x$ and $\rho_x$ correspond to an input classical state and its output quantum state, respectively.
Now, a data analyst needs to measure a quantum state by a measurement $(M_y)_{y=1}^m$ 
in order to obtain some information about the quantum state.
Hence, following the classical definition of DP, 
we impose the following condition on the quantum states $\rho_x$: 
\begin{equation}
	\forall(M_y)_{y=1}^m\text{ POVM},\ \text{the c.p.d.\ $\mathbb{P}(y|x)=\Tr\rho_xM_y$ satisfies $\epsilon$-DP},
	\label{CQ-DP}
\end{equation}
where ``c.p.d.''\ is an abbreviation of ``conditional probability distribution''; 
a quantum state and a POVM are briefly explained below.
\begin{itemize}
	\item
	A quantum state is defined as a \textit{density matrix}, i.e., a positive semi-definite matrix with trace one.
	\item
	An $m$-tuple $(M_i)_{i=1}^m$ of positive semi-definite matrices is called a \textit{positive-operator-valued measure (POVM)} 
	if the sum of all $M_i$ is equal to the identity matrix.
	A POVM is regarded as a measurement in quantum information theory.
	\item
	Given a quantum state $\rho$ and a measurement $(M_i)_{i=1}^m$, 
	the probability of obtaining each outcome $i=1,\ldots,m$ is $\Tr\rho M_i$.
\end{itemize}
Condition~\eqref{CQ-DP} is called \textit{classical-quantum $\epsilon$-differential privacy (CQ $\epsilon$-DP)} \cite{ISIT2020}, 
and is equivalent to the following one: 
\[
\forall x,x',\ \rho_{x'}\le e^\epsilon\rho_x,
\]
where for Hermitian matrices $H$ and $H'$ the inequality $H\le H'$ means for $H'-H$ to be positive semi-definite. 
The definition of CQ $\epsilon$-DP is a simple extension of the classical one, 
because \eqref{C-DP} can be written as 
\[
\forall x,x',\ p_{x'}\le e^\epsilon p_x
\]
if replacing the probability distributions $\mathbb{P}_{Y|X}(\cdot|x)$ with probability vectors $p_x$, 
where for probability vectors $p$ and $p'$ the inequality $p\le p'$ means for $p'-p$ to be non-negative.
From now on, we use an $n$-tuple $(p_i)_{i=1}^n$ of probability vectors instead of $(\mathbb{P}_{Y|X}(\cdot|x))_{x=1}^n$.

We summarize the above definitions.

\begin{definition}[Classical $\epsilon$-DP \cite{Duchi} and classical-quantum $\epsilon$-DP \cite{ISIT2020}]
	Let $\epsilon>0$ be a real number and $n\ge2$ be an integer.
	An $n$-tuple $(p_i)_{i=1}^n$ of probability vectors is called $\epsilon$-differentially private ($\epsilon$-DP) 
	if $p_i\le e^\epsilon p_j$ for all $i,j=1,\ldots,n$.
	An $n$-tuple $(\rho_i)_{i=1}^n$ of density matrices is called classical-quantum $\epsilon$-differentially private (CQ $\epsilon$-DP) 
	if $\rho_i\le e^\epsilon\rho_j$ for all $i,j=1,\ldots,n$.
	Also, define the sets $\C_n^{(d)}(\epsilon)$, $\CQ_n^{(d)}(\epsilon)$, $\C_n(\epsilon)$ and $\CQ_n(\epsilon)$ as 
	\begin{gather*}
		\C_n^{(d)}(\epsilon) = \{ \text{$\epsilon$-DP }(p_i)_{i=1}^n : \text{all $p_i$ are probability vectors in $\mathbb{R}^d$} \} \quad (d\ge2),\\
		\CQ_n^{(d)}(\epsilon) = \{ \text{CQ $\epsilon$-DP }(\rho_i)_{i=1}^n : \text{all $\rho_i$ are density matrices on $\mathbb{C}^d$} \} \quad (d\ge2),\\
		\C_n(\epsilon) = \bigcup_{d\ge2} \C_n^{(d)}(\epsilon),\quad
		\CQ_n(\epsilon) = \bigcup_{d\ge2} \CQ_n^{(d)}(\epsilon).
	\end{gather*}
\end{definition}

If $(\rho_i)_{i=1}^n$ is CQ $\epsilon$-DP, 
then all $\rho_i$ have the same support, i.e., 
all the ranges of $\rho_i$ are equal to one another.
Hence, we often implicitly assume that all $\rho_i$ have full rank 
if $(\rho_i)_{i=1}^n$ is CQ $\epsilon$-DP.

\vskip1ex
\textbf{Embedding classical states into quantum ones.}---%
Next, let us consider a subset of $\CQ_n(\epsilon)$ that corresponds to $\C_n(\epsilon)$.
For a probability vector $p=(p(i))_{i=1}^d\in\mathbb{R}^d$, 
define $\diag(p)$ as the diagonal matrix with diagonal entries $p(1),\ldots,p(d)$, 
which is a density matrix on $\mathbb{C}^d$.
Since a quantum (resp.\ classical) state is a density matrix (resp.\ probability vector), 
the mapping $\diag(\cdot)$ is an embedding from the set of classical states into the set of quantum ones.
Using the mapping $\diag(\cdot)$, 
we obtain the set 
\[
\diag(\C_n(\epsilon)) \coloneqq \{ (\diag(p_i))_{i=1}^n : (p_i)_{i=1}^n\in\C_n(\epsilon) \}
\]
that corresponds to $\C_n(\epsilon)$.

\vskip1ex
\textbf{Essentially classical elements.}---%
The set $\diag(\C_n(\epsilon))$ is much smaller than $\CQ_n(\epsilon)$, 
but actually, there is a set larger than $\diag(\C_n(\epsilon))$ that is ``essentially classical''.
To describe such a set, we consider two optimization problems: one is the classical case 
\begin{equation*}
	S_n^\C(\epsilon; \Phi) =
	\underbrace{\sup_{(p_i)_{i=1}^n\in\C_n(\epsilon)}}_{\text{Privacy protection}} \underbrace{\Phi(\diag(p_1),\ldots,\diag(p_n))}_{\text{Utility}},
\end{equation*}
which is often considered in information-theoretic studies of DP \cite{Kairouz,Holohan1, Geng1,Geng2,Geng3, ISIT2020}; 
the other is the quantum case 
\[
S_n^\CQ(\epsilon; \Phi) = \underbrace{\sup_{(\rho_i)_{i=1}^n\in\CQ_n(\epsilon)}}_{\text{Privacy protection}} \underbrace{\Phi(\rho_1,\ldots,\rho_n)}_{\text{Utility}}.
\]
The above $\Phi$ is a real-valued function of $n$ density matrices that represents the utility of private data, 
and the conditions $(p_i)_{i=1}^n\in\C_n(\epsilon)$ and $(\rho_i)_{i=1}^n\in\CQ_n(\epsilon)$ represent the privacy protection.
Since the data analyst's purpose is to maximally utilize private data under the privacy protection, 
we arrive at the above optimization problems.

Now, we want to define a subset of $\CQ_n(\epsilon)$ that is ``essentially classical''.
For this purpose, assume that the objective function $\Phi$ must satisfy monotonicity for \textit{completely positive and trace-preserving linear maps (CPTP maps)}; 
for the definition of CPTP maps, see Appendix.

\begin{definition}[Monotonicity for CPTP maps]\label{mono}
	A real-valued function $\Phi$ of $n$ density matrices is called monotone for CPTP maps if 
	\[
	\Phi(\Lambda(\rho_1),\ldots,\Lambda(\rho_n)) \le \Phi(\rho_1,\ldots,\rho_n)
	\]
	for all density matrices $\rho_1,\ldots,\rho_n$ and CPTP maps $\Lambda$.
	This inequality is called the data processing inequality (or information processing inequality).
\end{definition}

Since a CPTP map is regarded as a quantum operation in quantum information theory, 
information-theoretic quantities usually satisfy monotonicity for CPTP maps.
For example, quantum relative entropy, 
symmetric logarithmic derivative (SLD) Fisher information, 
Kubo--Mori--Bogoljubov (KMB) Fisher information, 
right logarithmic derivative (RLD) Fisher information, 
and trace distance satisfy monotonicity for CPTP maps \cite[Theorems~5.7 and 6.2]{book1}, \cite[Theorem~6.7 and Lemma~6.9]{book2}.

By monotonicity for CPTP maps, it follows that 
\[
\sup_{\substack{(p_i)_{i=1}^n\in\C_n(\epsilon) \\ \Lambda\text{ CPTP map}}} \Phi\bigl( \Lambda(\diag(p_1)),\ldots,\Lambda(\diag(p_n)) \bigr)
\le S_n^\C(\epsilon; \Phi).
\]
Moreover, the opposite inequality also holds 
since the identity mapping on $\Herm(d)$ is a CPTP map, 
where $\Herm(d)$ denotes the set of all Hermitian matrices on $\mathbb{C}^d$.
This fact leads us to the following definition.

\begin{definition}[Essentially classical element]\label{EC}
	Let $\epsilon>0$ be a real number and $n\ge2$ be an integer.
	We say that $(\rho_i)_{i=1}^n\in\CQ_n(\epsilon)$ is essentially classical 
	if there exist $(p_i)_{i=1}^n\in\C_n(\epsilon)$ and a CPTP map $\Lambda$ such that 
	$\Lambda(\diag(p_i))=\rho_i$ for all $i=1,\ldots,n$.
	We denote by $\EC_n(\epsilon)$ the set of all essentially classical elements in $\CQ_n(\epsilon)$.
\end{definition}

Although an element in $\EC_n(\epsilon)$ consists of quantum states, 
the equality 
\[
S_n^\EC(\epsilon; \Phi)=S_n^\C(\epsilon; \Phi)
\]
holds, where $S_n^\EC(\epsilon; \Phi)$ is defined in the same way as $S_n^\CQ(\epsilon; \Phi)$.
Hence, the comparison of $S_n^\C(\epsilon; \Phi)$ and $S_n^\CQ(\epsilon; \Phi)$ 
is the same as that of $S_n^\EC(\epsilon; \Phi)$ and $S_n^\CQ(\epsilon; \Phi)$.

\vskip1ex
\textbf{Comparison of $\EC_n(\epsilon)$ and $\CQ_n(\epsilon)$.}---%
Although the set $\EC_n(\epsilon)$ is a subset of $\CQ_n(\epsilon)$, 
we are interested in whether they are equal to each other or not.
If $\EC_n(\epsilon)$ is equal to $\CQ_n(\epsilon)$, 
then $S_n^\C(\epsilon; \Phi)=S_n^\EC(\epsilon; \Phi)=S_n^\CQ(\epsilon; \Phi)$, 
i.e., CQ $\epsilon$-DP mechanisms have no quantum advantage in optimization.
In this perspective, it is important to compare $\EC_n(\epsilon)$ with $\CQ_n(\epsilon)$.
The following fact follows from \cite[Theorem~1]{ISIT2020}.

\begin{proposition}\label{n=2}
	For all $\epsilon>0$, $\EC_2(\epsilon)=\CQ_2(\epsilon)$.
\end{proposition}

By Proposition~\ref{n=2}, it follows that $S_2^\C(\epsilon; \Phi)=S_2^\EC(\epsilon; \Phi)=S_2^\CQ(\epsilon; \Phi)$.
However, it turns out that $\EC_n(\epsilon)\not=\CQ_n(\epsilon)$ for every $n\ge3$ (Corollary~\ref{main1}).
Hence, using the following definition, 
we investigate the difference between $\EC_n(\epsilon)$ and $\CQ_n(\epsilon)$.

\begin{definition}\label{En}
	For $\epsilon>0$ and $n\ge2$, define the set $\cE_n(\epsilon)$ as 
	\[
	\cE_n(\epsilon) = \{ \epsilon'>0 : \CQ_n(\epsilon)\subset\EC_n(\epsilon') \}.
	\]
\end{definition}

Actually, $\cE_n(\epsilon)$ is non-empty (Theorem~\ref{main2}).
Since $\EC_n(\epsilon)$ is monotonically increasing in $\epsilon>0$, 
the set $\cE_n(\epsilon)$ is an interval that is not bounded above, i.e., 
of the form $[\epsilon_{\inf},\infty)$ or $(\epsilon_{\inf},\infty)$.
Since $\EC_n(\epsilon)$ is a subset of $\CQ_n(\epsilon)$, 
and since $\CQ_n(\epsilon)$ is strictly increasing in $\epsilon>0$, 
every $\epsilon'\in\cE_n(\epsilon)$ is greater than or equal to $\epsilon$.
Therefore, Proposition~\ref{n=2} implies that $\cE_2(\epsilon)=[\epsilon,\infty)$.
We estimate the infimum of $\cE_n(\epsilon)$.

\vskip1ex
\textbf{Main results.}---%
In this paper, we show the following theorems.

\begin{theorem}\label{main2}
	For all $\epsilon>0$ and $n\ge2$, 
	there exists $\epsilon'\in\cE_n(\epsilon)$ such that $e^{\epsilon'}-1=(n-1)(e^\epsilon-1)$.
\end{theorem}

\begin{theorem}\label{main3}
	For all $n\ge2$, $\epsilon>0$ and $\epsilon'\in\cE_n(\epsilon)$, 
	\[
	\frac{e^{\epsilon'}-1}{e^\epsilon-1}
	\ge F_n(e^\epsilon-1),
	\]
	where $F_n$ is defined as follows: 
	\begin{align*}
		g_n(t) &=
		\begin{dcases}
			\frac{2}{t-1}(\sqrt{(n-1)(n-t)}+n-t) & 1<t\le n,\\
			\infty & t=1,
		\end{dcases}
		\\
		a_{n,t}(x) &= t\Bigl( \frac{n-t}{n-1}(x+2)^2 - x^2 \Bigr) \quad (1\le t\le n,\ x\ge0),\\
		f_{n,k}(x) &= \frac{(n+2k)x + \sqrt{(n+2k)^2x^2 + 8na_{n,k}(x)}}{2a_{n,k}(x)} \quad (1\le k\le n/2,\ 0\le x<g_n(k)),\\
		F_n(x) &= \min\{ f_{n,k}(x) : 1\le k\le n/2\text{ with }x<g_n(k) \} \quad (x\ge0).
	\end{align*}
	For several properties of the functions $g_n$, $a_{n,t}$, $f_{n,k}$ and $F_n$, 
	see Lemmas~$\ref{lem4}$ and $\ref{lem5}$.
\end{theorem}

Theorem~\ref{main2} with $n=2$ implies Proposition~\ref{n=2}.
Moreover, by Theorems~\ref{main2} and \ref{main3}, 
the infimum $\epsilon_{\inf}=\epsilon_{\inf}(n,\epsilon)=\inf\cE_n(\epsilon)$ satisfies that 
\begin{equation}
	F_n(e^\epsilon-1) \le \frac{e^{\epsilon_{\inf}}-1}{e^\epsilon-1} \le n-1.
	\label{eq01}
\end{equation}
Lemma~\ref{lem5}, which is proved in Section~\ref{proof3}, yields that 
\begin{itemize}
	\item
	for all $x\ge0$, $F_2(x)=1$;
	\item
	for all $n\ge3$, $F_n$ is strictly increasing;
	\item
	for all $n\ge2$, 
	\[
	F_n(0) = \sqrt{\frac{n(n-1)}{2\lfloor{n/2}\rfloor\lceil{n/2}\rceil}}\quad\text{and}\quad
	\lim_{x\to\infty} F_n(x) = \frac{n+2}{4},
	\]
	where $\lfloor{x}\rfloor$ (resp.\ $\lceil{x}\rceil$) denotes the greatest (resp.\ least) integer $\le x$ (resp.\ $\ge x$) 
	for a real number $x$.
\end{itemize}
Since $F_n(x)>F_n(0)>1$ for all $n\ge3$ and $x>0$, 
we obtain the following corollary.

\begin{corollary}\label{main1}
	For all $\epsilon>0$ and $n\ge3$, 
	$\EC_n(\epsilon)\not=\CQ_n(\epsilon)$.
\end{corollary}

Since we give a concrete objective function $\Phi$ such that 
$S_n^\C(\epsilon; \Phi)=S_n^\EC(\epsilon; \Phi)<S_n^\CQ(\epsilon; \Phi)$ for every $n\ge3$ (Theorem~\ref{main0}), 
Corollary~\ref{main1} also follows from Theorem~\ref{main0} 
(although Theorem~\ref{main3} is proved by using Theorem~\ref{main0}).
Theorem~\ref{main0} implies a sufficient condition 
for a CQ $\epsilon$-DP $n$-tuple not to lie in $\EC_n(\epsilon)$ (Corollary~\ref{main4}).
Using Corollary~\ref{main4}, 
we construct CQ $\epsilon$-DP $n$-tuples that do not lie in $\EC_n(\epsilon)$ (Section~\ref{concrete}).

We mention a relation among this paper and existing studies briefly.
Ref.\ \cite{ISIT2020} handles the classical-quantum setting as well as this paper, 
but Refs.\ \cite{qDP1,qDP2,qDP3,qDP4} consider the case when input and output states are quantum.
The definition of CQ $\epsilon$-DP can be regarded as a special case of quantum DP \cite{qDP3}, 
but \cite{qDP3} does not include our results.

\vskip1ex
\textbf{Supplement on the set $\EC_n(\epsilon)$.}---%
Actually, the set $\EC_n(\epsilon)$ can be written without CPTP maps.

\begin{proposition}\label{CPTP}
	For all $\epsilon>0$ and $n\ge2$, 
	\begin{equation}
		\EC_n(\epsilon) = \biggl\{ \Bigl( \sum_k p_i(k)\sigma_k \Bigr)_{i=1}^n : (p_i)_{i=1}^n\in\C_n(\epsilon),\ \text{density matrices }\sigma_k \biggr\},
		\label{eq-CPTP}
	\end{equation}
	where the above sum is taken all over $k=1,\ldots,d$ if $d$ is the dimension of the vector space that $p_1,\ldots,p_n$ inhabit.
\end{proposition}

Proposition~\ref{CPTP} can easily be checked; see Appendix.
Although we have defined the set $\EC_n(\epsilon)$ with CPTP maps, 
the same set is obtained even if replacing CPTP maps with positive and trace-preserving linear maps (PTP maps).
That is, complete positivity is unnecessary, and positivity suffices in Definition~\ref{EC}.
However, we have used CPTP maps in Definition~\ref{EC} 
because CPTP maps are more natural in quantum information theory than PTP maps, 
and monotonicity for CPTP maps is used in Section~\ref{result}.

\section{Another main result}\label{result}

In this section, we state another main result (Theorem~\ref{main0}), 
which is proved in Section~\ref{proof0}.
Theorem~\ref{main0} asserts that a certain objective function $\Phi$ satisfies that 
$S_n^\C(\epsilon; \Phi)=S_n^\EC(\epsilon; \Phi)<S_n^\CQ(\epsilon; \Phi)$ for all $\epsilon>0$ and $n\ge3$.
The objective function $\Phi$ in Theorem~\ref{main0} is constructed 
by using the RLD Fisher information of a one-parameter family.

\begin{definition}[RLD Fisher information {\cite[p.~260]{book1}}]
	For density matrices $\rho$ and $\sigma$ with full rank, 
	we denote the RLD Fisher information of the one-parameter family $((1-\theta)\rho+\theta\sigma)_{\theta\in[0,1]}$ 
	at the point $\theta$ as 
	\[
	J_\theta(\rho,\sigma) = \Tr(\sigma-\rho)^2((1-\theta)\rho+\theta\sigma)^{-1}.
	\]
	For probability vectors $p$ and $q$, 
	we set $J_\theta(p,q)=J_\theta(\diag(p),\diag(q))$.
\end{definition}

The function $J_\theta$ satisfies monotonicity for CPTP maps (see also Definition~\ref{mono}).
If $(\rho_i)_{i=1}^n$ is CQ $\epsilon$-DP, 
we may assume that all $\rho_i$ have full rank (see Section~\ref{intro}), 
and hence, we can consider the value $J_\theta(\rho_i,\rho_j)$ for all $i,j=1,\ldots,n$.
Also, for probability vectors $p$ and $q$, 
the value $J_\theta(p,q)$ is the Fisher information in the classical sense.
From now on, we denote by $\avg_{i\not=j} \alpha_{i,j}$ the arithmetic mean of real numbers $\alpha_{i,j}$, $i\not=j$.

\begin{theorem}\label{main0}
	For real numbers $\theta\in[0,1]$ and $\epsilon>0$ and an integer $n\ge2$, 
	we define the suprema $M_n^\C(\epsilon;J_\theta)$, $M_n^\EC(\epsilon;J_\theta)$ and $M_n^\CQ(\epsilon;J_\theta)$ as 
	\begin{align*}
		M_n^\C(\epsilon;J_\theta) &= \sup_{(p_i)_{i=1}^n\in\C_n(\epsilon)} \min_{i\not=j} J_\theta(p_i,p_j),\\
		M_n^{\mathrm{X}}(\epsilon;J_\theta) &= \sup_{(\rho_i)_{i=1}^n\in\mathrm{X}_n(\epsilon)} \min_{i\not=j} J_\theta(\rho_i,\rho_j)
		\quad(\mathrm{X}=\EC,\CQ).
	\end{align*}
	Then, for all $\theta\in[0,1]$, $\epsilon>0$ and $n\ge2$, 
	we have $M_n^\CQ(\epsilon;J_\theta) = M_2^\C(\epsilon;J_\theta)$ and 
	\begin{align*}
		M_n^\EC(\epsilon;J_\theta) &= M_n^\C(\epsilon;J_\theta)
		= \sup_{(p_i)_{i=1}^n\in\C_n(\epsilon)} \avg_{i\not=j} J_\theta(p_i,p_j)\\
		&= \frac{f_\theta(e^\epsilon,1)+f_\theta(1,e^\epsilon)}{n-1}\max_{1\le k\le n/2} \frac{k(n-k)}{ke^\epsilon+n-k},
	\end{align*}
	where $f_\theta(\alpha,\beta) \coloneqq (\alpha-\beta)^2/((1-\theta)\alpha+\theta\beta)$ for $\alpha,\beta>0$.
	Moreover, $M_n^\EC(\epsilon;J_\theta) < M_n^\CQ(\epsilon;J_\theta)$ for all $\theta\in[0,1]$, $\epsilon>0$ and $n\ge3$.
\end{theorem}

If we set $\Phi(\rho_1,\ldots,\rho_n)=\min_{i\not=j} J_\theta(\rho_i,\rho_j)$, 
then $S_n^{\mathrm{X}}(\epsilon; \Phi)=M_n^{\mathrm{X}}(\epsilon; J_\theta)$ for $\mathrm{X}=\C,\EC,\CQ$.
Hence, Theorem~\ref{main0} gives us a concrete objective function $\Phi$ such that 
$S_n^\C(\epsilon; \Phi)=S_n^\EC(\epsilon; \Phi)<S_n^\CQ(\epsilon; \Phi)$ for all $\epsilon>0$ and $n\ge3$.
Moreover, Theorem~\ref{main0} implies Corollary~\ref{main1} and the following corollary immediately.

\begin{corollary}\label{main4}
	Let $\epsilon>0$ be a real number and $n\ge3$ be an integer.
	If $(\rho_i)_{i=1}^n\in\CQ_n(\epsilon)$ satisfies that 
	$M_n^\C(\epsilon;J_\theta) < \avg_{i\not=j} J_\theta(\rho_i,\rho_j)$ for some $\theta\in[0,1]$, 
	then $(\rho_i)_{i=1}^n$ does not lie in $\EC_n(\epsilon)$.
\end{corollary}

\section{Proof of Theorem~$\ref{main2}$}\label{proof2}

Denote by $I_d$ the identity matrix of order $d$, 
and by $\mathbf{1}_d$ the column vector $[1,\ldots,1]^\top\in\mathbb{R}^d$.
We often use the bra-ket notation: for $u\in\mathbb{C}^d$, 
$\ket{u}$ and $\bra{u}$ denote the column vector $u$ and its conjugate transpose, respectively.
Hence, $\braket{\cdot|\cdot}$ is the standard Hermitian inner product on $\mathbb{C}^d$, 
and $\ketbra{u}{u}$ is a rank-one orthogonal projection for every unit vector $u\in\mathbb{C}^d$.

To prove Theorem~\ref{main2}, we begin with the following preliminary lemma.

\begin{lemma}\label{5.1-lem1}
	If density matrices $\rho_1,\ldots,\rho_n$ on $\mathbb{C}^d$ are orthogonal to each other, i.e., $\Tr\rho_i\rho_j=0$ for all $i\not=j$, 
	then for all density matrices $\sigma_1,\ldots,\sigma_n$ on $\mathbb{C}^{d'}$ 
	there exists a CPTP map $\Lambda$ such that 
	$\Lambda(\rho_i)=\sigma_i$ for all $i=1,\ldots,n$.
\end{lemma}
\begin{proof}
	Let $\rho_1,\ldots,\rho_n$ be density matrices that are orthogonal to each other.
	For $i=1,\ldots,n$, take the orthogonal projection $P_i$ onto the support of $\rho_i$.
	Put $P_0 = I_d - \sum_{i=1}^n P_i$.
	Then $P_0,P_1,\ldots,P_n$ are also orthogonal to each other.
	Defining $\Lambda(X) = \sum_{i=0}^n (\Tr XP_i)\sigma_i$ for $X\in\Herm(d)$, 
	we find that $\Lambda$ is a CPTP map satisfying that 
	$\Lambda(\rho_i)=\sigma_i$ for all $i=1,\ldots,n$.
\end{proof}

\begin{proof}[Proof of Theorem~$\ref{main2}$]
	Let $\epsilon>0$ be a real number and $n\ge2$ be an integer.
	Define $\epsilon'>0$ as $e^{\epsilon'}-1=(n-1)(e^\epsilon-1)$, 
	and the $n$-tuple $(p_i)_{i=1}^n$ of probability vectors in $\mathbb{R}^n$ as 
	\[
	[p_1,\ldots,p_n]
	= \frac{(e^{\epsilon'}-1)I_n + \mathbf{1}_n\mathbf{1}_n^\top}{e^{\epsilon'}+n-1}
	= \frac{1}{e^{\epsilon'}+n-1}
	\begin{bmatrix}
		e^{\epsilon'}&1&\cdots&1\\
		1&e^{\epsilon'}&\ddots&\vdots\\
		\vdots&\ddots&\ddots&1\\
		1&\cdots&1&e^{\epsilon'}
	\end{bmatrix}.
	\]
	Then $(p_i)_{i=1}^n$ lies in $\C_n(\epsilon')$.
	Moreover, for every $k=1,\ldots,n$, 
	\begin{equation}
		(e^{\epsilon'}+n-2)p_k - \sum_{i\not=k} p_i
		= \frac{(e^{\epsilon'}+n-2)e^{\epsilon'} - (n-1)}{e^{\epsilon'}+n-1}e_k
		= (e^{\epsilon'}-1)e_k,
		\label{eq02}
	\end{equation}
	where $(e_i)_{i=1}^n$ denotes the standard basis of $\mathbb{R}^n$.
	\par
	Now, let $(\rho_i)_{i=1}^n$ be CQ $\epsilon$-DP.
	We show that $(\rho_i)_{i=1}^n$ lies in $\EC_n(\epsilon')$.
	By the definitions of CQ $\epsilon$-DP and $\epsilon'$, for every $k=1,\ldots,n$, 
	\begin{align*}
		(e^{\epsilon'}+n-2)\rho_k - \sum_{i\not=k} \rho_i
		&\ge (e^{\epsilon'}+n-2)\rho_k - \sum_{i\not=k} e^\epsilon\rho_k\\
		&= \bigl( e^{\epsilon'}+n-2 - (n-1)e^\epsilon \bigr)\rho_k = 0.
	\end{align*}
	Thus, the left-hand side can be rewritten as 
	\[
	(e^{\epsilon'}+n-2)\rho_k - \sum_{i\not=k} \rho_i
	= (e^{\epsilon'}-1)\sigma_k,
	\]
	with density matrices $\sigma_1,\ldots,\sigma_n$.
	By Lemma~\ref{5.1-lem1}, there exists a CPTP map $\Lambda$ such that 
	$\Lambda(\ketbra{e_k}{e_k})=\sigma_k$ for every $k=1,\ldots,n$.
	This and \eqref{eq02} yield that for every $k=1,\ldots,n$, 
	\begin{align*}
		(e^{\epsilon'}+n-2)\Lambda(\diag(p_k)) - \sum_{i\not=k} \Lambda(\diag(p_i))
		&= (e^{\epsilon'}-1)\Lambda(\diag(\ketbra{e_k}{e_k}))
		= (e^{\epsilon'}-1)\sigma_k\\
		&= (e^{\epsilon'}+n-2)\rho_k - \sum_{i\not=k} \rho_i.
	\end{align*}
	Solving the above simultaneous equations, 
	we obtain $\Lambda(\diag(p_k))=\rho_k$ for every $k=1,\ldots,n$.
	This implies that $(\rho_i)_{i=1}^n$ lies in $\EC_n(\epsilon')$.
\end{proof}

\section{Proof of Theorem~$\ref{main3}$}\label{proof3}

In this section, assuming Theorem~\ref{main0}, we prove Theorem~\ref{main3}.
First, we begin with several lemmas on the functions $g_n$, $a_{n,t}$, $f_{n,k}$ and $F_n$ in Theorem~\ref{main3}.
These lemmas are necessary to prove Theorem~\ref{main3}.

\begin{lemma}\label{lem4}
	For $g_n$ and $a_{n,t}$ in Theorem~$\ref{main3}$, the following facts hold.
	\begin{enumerate}
		\item
		For all $n\ge2$, $g_n$ is strictly decreasing.
		\item
		For all $n\ge2$, $1\le t<n$ and $x\ge0$, 
		the inequality $x<g_n(t)$ is equivalent to $a_{n,t}(x)>0$.
		\item
		For all $n\ge2$, $1<t<n$ and $x\ge0$, 
		the equality $x=g_n(t)$ is equivalent to $a_{n,t}(x)=0$.
	\end{enumerate}
\end{lemma}
\begin{proof}
	Proof of fact~1. Trivial.
	\par
	Proof of fact~2.
	If $t=1$, the assertion is trivial.
	For all $n\ge2$, $1<t<n$ and $x\ge0$, we have 
	\begin{align*}
		a_{n,t}(x)>0
		&\iff \sqrt{\frac{n-t}{n-1}}(x+2) > x
		\iff 2\sqrt{\frac{n-t}{n-1}} > \Bigl( 1-\sqrt{\frac{n-t}{n-1}} \Bigr)x\\
		&\iff 2\sqrt{\frac{n-t}{n-1}}\Bigl( 1+\sqrt{\frac{n-t}{n-1}} \Bigr) > \Bigl( 1-\frac{n-t}{n-1} \Bigr)x = \frac{t-1}{n-1}x\\
		&\iff x < \frac{2}{t-1}( \sqrt{(n-1)(n-t)}+n-t ) = g_n(t).
	\end{align*}
	Therefore, fact~2 holds.
	\par
	Proof of fact~3.
	Due to $1<t<n$, fact~3 can be proved in a similar way to fact~2.
\end{proof}

\begin{lemma}\label{lem5}
	For $g_n$, $a_{n,t}$, $f_{n,k}$ and $F_n$ in Theorem~$\ref{main3}$, the following facts hold.
	\begin{enumerate}
		\item
		For all $n\ge2$, $1\le k\le n/2$ and $0\le x<g_n(k)$, 
		the quadratic equation 
		\begin{equation}
			\frac{k(n-k)}{n-1}(x+2)^2y^2 = (xy+2)(kxy+n)
			\label{eq03}
		\end{equation}
		for $y$ has the unique positive solution $y=f_{n,k}(x)$, 
		and the other solution is negative.
		\item
		For all $x\ge0$, $F_2(x)=1$.
		\item
		For all $n\ge3$ and $1\le k\le n/2$, $f_{n,k}(x)$ is strictly increasing in $0\le x<g_n(k)$.
		\item
		For all $n\ge3$, $F_n$ is strictly increasing.
		\item
		For all $n\ge2$, 
		\[
		F_n(0) = \sqrt{\frac{n(n-1)}{2\lfloor{n/2}\rfloor\lceil{n/2}\rceil}}\quad\text{and}\quad
		\lim_{x\to\infty} F_n(x) = \frac{n+2}{4}.
		\]
		\item
		For all $n\ge3$ and $1<t\le n/2$, 
		the quadratic equation 
		\[
		p_{n,t}(x) \coloneqq a_{n,t}(x)(n+2)^2/4 - (n+2t)(n+2)x - 8n = 0
		\]
		for $x$ has a unique positive solution $x=G_n(t)$, 
		and the other solution is negative.
		Moreover, $G_n(t)<g_n(t)$ for all $n\ge3$ and $1<t\le n/2$.
		\item
		For all $n\ge3$ and $x\ge0$, 
		\begin{equation*}
			F_n(x) = \min\{ f_{n,k}(x) : 1\le k\le n/2\text{ with }x<G_n(k) \},
		\end{equation*}
		where $G_n(1) \coloneqq \infty$.
		\item
		For all $n\ge4$ and $2\le k\le n/2$, $G_n(k)<2n-6$.
	\end{enumerate}
\end{lemma}
\begin{proof}
	Proof of fact~1.
	Let $n\ge2$ and $1\le k\le n/2$ be integers, and $0\le x<g_n(k)$ be a real number.
	Eq.~\eqref{eq03} can be rewritten as 
	\begin{align}
		\text{\eqref{eq03}}
		&\iff k\Bigl( \frac{n-k}{n-1}(x+2)^2 - x^2 \Bigr)y^2 - (n+2k)xy - 2n = 0\nonumber\\
		&\iff a_{n,k}(x)y^2 - (n+2k)xy - 2n = 0. \label{eq04}
	\end{align}
	Since $a_{n,k}(x)>0$ by fact~2 of Lemma~\ref{lem4}, fact~1 follows.
	\par
	Proof of fact~2.
	Since $F_2(x)=f_{2,1}(x)$, fact~2 follows from the definition of $f_{2,1}$.
	\par
	Proof of fact~3.
	Let $n\ge3$ and $1\le k\le n/2$ be integers.
	We show that $f_{n,k}$ is strictly increasing.
	Put $y=f_{n,k}(x)$ with $0\le x<g_n(k)$.
	Differentiating both sides in \eqref{eq04} with respect to $x$, we have 
	\[
	a'_{n,k}(x)y^2 + 2a_{n,k}(x)yy' - (n+2k)y - (n+2k)xy'= 0,
	\]
	which yields that 
	\begin{equation}
		(2a_{n,k}(x)y - (n+2k)x)y^{-1}y' = n+2k - a'_{n,k}(x)y. \label{eq05}
	\end{equation}
	Since $y>(n+2k)x/a_{n,k}(x)$, 
	if the right-hand side in \eqref{eq05} is positive, then $y'$ is also positive.
	Write the right-hand side in \eqref{eq05} as $h(x)$.
	Since 
	\begin{equation}
		f_{n,k}(0) = \frac{\sqrt{8na_{n,k}(0)}}{2a_{n,k}(0)}
		= \sqrt{\frac{2n}{a_{n,k}(0)}}
		= \sqrt{\frac{n(n-1)}{2k(n-k)}},
		\label{eq06}
	\end{equation}
	it turns out that 
	\begin{align*}
		h(0) &= n+2k - a'_{n,k}(0)f_{n,k}(0)
		= n+2k - \frac{4k(n-k)}{n-1}\sqrt{\frac{n(n-1)}{2k(n-k)}}\\
		&= n+2k - \sqrt{\frac{8nk(n-k)}{n-1}}
		= (\sqrt{n}-\sqrt{2k})^2 + \sqrt{8nk}\Bigl( 1 - \sqrt{\frac{n-k}{n-1}} \Bigr)
		> 0,
	\end{align*}
	where the assumption $n\ge3$ has been used to obtain the last inequality.
	We show that $h(x)>0$ by contradiction.
	Suppose that there exists $0<x_1<g_n(k)$ such that $h(x_1)\le0$.
	By the continuity of $h$, we can take the minimum value $0<x_2\le x_1$ such that $h(x_2)\le0$.
	It is clear that $h(x_2)=0$ (if not so, then $x_2$ would not be the minimum).
	This implies that $f'_{n,k}(x_2)=0$, since $h(x)$ is the right-hand side in \eqref{eq05}.
	Also, we find that $h(x_2-\delta)>0$ for all $0<\delta<x_2$, and thus, 
	\[
	h'(x_2) = \lim_{\delta\to+0} \frac{h(x_2-\delta)-h(x_2)}{-\delta}
	= \lim_{\delta\to+0} \frac{h(x_2-\delta)}{-\delta}
	\le 0.
	\]
	However, 
	\[
	h'(x_2) = \frac{2k(k-1)}{n-1}f_{n,k}(x_2) > 0,
	\]
	since $f'_{n,k}(x_2)=0$, $f_{n,k}(x_2)>0$, and 
	\[
	h'(x) = -a''_{n,k}(x)y - a'_{n,k}(x)y'
	= \frac{2k(k-1)}{n-1}y - a'_{n,k}(x)y'.
	\]
	By this contradiction, we conclude that $h(x)>0$ and thus, $y'>0$.
	Therefore, $f_{n,k}$ is strictly increasing.
	\par
	Proof of fact~4.
	Let $n\ge3$ be an integer.
	We show that $F_n$ is strictly increasing.
	Let $x_2>x_1>0$.
	For $i=1,2$, define the integer $n_i$ as 
	\[
	n_i = \max\{ 1\le k\le n/2 : x_i<g_n(k) \}
	= \max\{ 1\le k\le n/2 : k<g_n^{-1}(x_i) \},
	\]
	where $g_n^{-1}$ is the inverse function of $g_n$.
	By fact~1 of Lemma~\ref{lem4}, $g_n^{-1}(x_2)<g_n^{-1}(x_1)$ and thus, $n_2\le n_1$.
	By this and fact~3, 
	\[
	F_n(x_1) = \min_{1\le k\le n_1} f_{n,k}(x_1)
	< \min_{1\le k\le n_2} f_{n,k}(x_2) = F_n(x_2).
	\]
	Therefore, $F_n$ is strictly increasing.
	\par
	Proof of fact~5.
	Let $n\ge2$ be an integer.
	From \eqref{eq06}, it follows that 
	\[
	F_n(0) = \min_{1\le k\le n/2} \sqrt{\frac{n(n-1)}{2k(n-k)}}
	= \sqrt{\frac{n(n-1)}{2\lfloor{n/2}\rfloor\lceil{n/2}\rceil}}.
	\]
	Since $x\ge g_n(k)$ for all $2\le k\le n/2$ and $x\ge g_n(2)$, 
	and since $a_{n,1}(x)=4(x+1)$, 
	we obtain 
	\[
	\lim_{x\to\infty} F_n(x)
	= \lim_{x\to\infty} f_{n,1}(x)
	= \lim_{x\to\infty} \frac{(n+2)x + \sqrt{(n+2)^2x^2 + 8na_{n,1}(x)}}{2a_{n,1}(x)}
	= \frac{n+2}{4}.
	\]
	\par
	Proof of fact~6.
	Let $n\ge3$ be an integer and $1<t\le n/2$ be a real number.
	Then 
	\begin{align*}
		p_{n,t}(0) &= a_{n,t}(0)(n+2)^2/4 - 8n
		= \frac{t(n-t)}{n-1}(n+2)^2 - 8n\\
		&> (n+2)^2 - 8n
		= (n-2)^2 > 0.
	\end{align*}
	Also, the leading coefficient of $p_{n,t}(x)$ is $t(1-t)(n+2)^2/4(n-1)<0$.
	Therefore, the quadratic equation $p_{n,t}(x)=0$ has a unique positive solution $x=G_n(t)$, 
	and the other solution is negative.
	By fact~3 of Lemma~\ref{lem4}, we have $a_{n,t}(g_n(t))=0$ and thus, $p_{n,t}(g_n(t))<0$.
	From this and $g_n(t)>0$, it follows that $G_n(t)<g_n(t)$.
	\par
	Proof of fact~7.
	Let $n\ge3$ be an integer.
	By facts~4 and 5, $F_n(x)<(n+2)/4$ for all $x\ge0$.
	Thus, for all $x\ge0$, 
	\begin{equation}
		F_n(x) = \min\{ f_{n,k}(x) : 1\le k\le n/2\text{ with }x<g_n(k)\text{ and }f_{n,k}(x)<(n+2)/4 \}.
		\label{eq07}
	\end{equation}
	For all $2\le k\le n/2$ and $0\le x<g_n(k)$, we have 
	\begin{align*}
		f_{n,k}(x) < (n+2)/4
		&\overset{(a)}{\iff} a_{n,k}(x)\Bigl( \frac{n+2}{4} \Bigr)^2 - (n+2k)x\cdot\frac{n+2}{4} - 2n > 0\\
		&\iff p_{n,k}(x) > 0
		\overset{(b)}{\iff} x < G_n(k),
	\end{align*}
	where $(a)$ is derived from fact~1, \eqref{eq04}, and $a_{n,k}(x)>0$ (fact~2 of Lemma~\ref{lem4}); 
	$(b)$ is derived from fact~6 (since the leading coefficient of $p_{n,k}(x)$ is $k(1-k)(n+2)^2/4(n-1)<0$).
	The above equivalence of $f_{n,k}(x) < (n+2)/4$ and $x < G_n(k)$ is also true for $k=1$, 
	since $f_{n,1}(x)<f_{n,1}(+\infty)=(n+2)/4$ (see the definition of $f_{n,1}$ and fact~3) and $G_n(1)=\infty$.
	Therefore, the right-hand side in \eqref{eq07} is equal to 
	\[
	\min\{ f_{n,k}(x) : 1\le k\le n/2\text{ with }x<G_n(k) \}.
	\]
	\par
	Proof of fact~8.
	Let $n\ge4$ and $2\le k\le n/2$ be integers.
	If $k\ge4$, then fact~6 of Lemma~\ref{lem5} and fact~1 of Lemma~\ref{lem4} imply that 
	\[
	G_n(k) < g_n(k) \le g_n(4)
	= \frac{2}{3}(\sqrt{(n-1)(n-4)} + n-4)
	< \frac{2}{3}(n-2 + n-4)
	< 2n-6.
	\]
	Consider the remaining cases $k=2,3$.
	Recalling the definition of $p_{n,t}$, we have 
	\[
	p_{n,k}(2n-6) = a_{n,k}(2n-6)(n+2)^2/4 - (n+2k)(n+2)(2n-6) - 8n.
	\]
	Since 
	\begin{align*}
		a_{n,2}(2n-6)
		&= 8\Bigl( \frac{n-2}{n-1}(n-2)^2 - (n-3)^2 \Bigr)\\
		&= 8\Bigl( \frac{n-2}{n-1}(n-1)(n-3) + \frac{n-2}{n-1} - (n-3)^2 \Bigr)\\
		&= 8\Bigl( (n-2)(n-3) + 1-\frac{1}{n-1} - (n-3)^2 \Bigr)
		= 8\Bigl( n-2 - \frac{1}{n-1} \Bigr)
	\end{align*}
	and 
	\begin{align*}
		a_{n,3}(2n-6)
		&= 12\Bigl( \frac{n-3}{n-1}(n-2)^2 - (n-3)^2 \Bigr)\\
		&= 12(n-3)\Bigl( \frac{(n-2)^2}{n-1} - (n-3) \Bigr)
		= \frac{12(n-3)}{n-1},
	\end{align*}
	we have 
	\begin{align*}
		p_{n,2}(2n-6)
		&= 2\Bigl( n-2 - \frac{1}{n-1} \Bigr)(n+2)^2 - 2(n+4)(n+2)(n-3) - 8n\\
		&= 2(n+2)\bigl( (n-2)(n+2) - (n+4)(n-3) \bigr) - \frac{2(n+2)^2}{n-1} - 8n\\
		&= 2(n+2)(8-n) - 2\cdot\frac{(n-1)(n+5)+9}{n-1} - 8n\\
		&< 2(n+2)(8-n) - 2(n+5) - 8n\\
		&= 2(n+2)(4-n) - 2(n-3) < 0
	\end{align*}
	and 
	\begin{align*}
		p_{n,3}(2n-6)
		&= \frac{3(n-3)}{n-1}(n+2)^2 - 2(n+6)(n+2)(n-3) - 8n\\
		&= (n-3)(n+2)\Bigl( \frac{3(n+2)}{n-1} - 2(n+6) \Bigr) - 8n\\
		&\le (n-3)(n+2)\bigl( 6 - 2(n+3) \bigr) - 8n\\
		&\le (n-3)(n+2)\cdot(-2)(n+3) - 8n < 0.
	\end{align*}
	Since the leading coefficient of $p_{n,k}(x)$ is $k(1-k)(n+2)^2/4(n-1)<0$ for $k=2,3$, 
	it follows from fact~6 and $2n-6>0$ that $G_n(2)<2n-6$ and $G_n(3)<2n-6$.
\end{proof}

Next, assuming Theorem~\ref{main0}, we prove Theorem~\ref{main3}.

\begin{proof}[Proof of Theorem~$\ref{main3}$]
	For $n=2$, the assertion follows from fact~2 of Lemma~\ref{lem5} (see also the sentences below Definition~\ref{En}).
	Let $n\ge3$ be an integer, $\epsilon>0$ be a real number, 
	and $\epsilon_1>0$ satisfy $e^{\epsilon_1}-1<(e^\epsilon-1)F_n(e^\epsilon-1)$.
	We show that $M_n^\CQ(\epsilon;J_{1/2}) > M_n^\EC(\epsilon_1;J_{1/2})$.
	Put $x=e^\epsilon-1$ and $y=(e^{\epsilon_1}-1)/x<F_n(x)$.
	By Theorem~\ref{main0}, 
	\begin{align*}
		M_n^\CQ(\epsilon;J_{1/2})
		&= M_2^\C(\epsilon;J_{1/2})
		= \frac{4(e^\epsilon-1)^2}{e^\epsilon+1}\cdot\frac{1}{e^\epsilon+1}
		= \frac{4x^2}{(x+2)^2},\\
		M_n^\EC(\epsilon_1;J_{1/2})
		&= M_n^\C(\epsilon_1;J_{1/2})
		= \frac{4(e^{\epsilon_1}-1)^2}{e^{\epsilon_1}+1}\cdot\frac{1}{n-1}\max_{1\le k\le n/2} \frac{k(n-k)}{ke^{\epsilon_1}+n-k}\\
		&= \frac{4(xy)^2}{(n-1)(xy+2)}\max_{1\le k\le n/2} \frac{k(n-k)}{kxy+n}.
	\end{align*}
	Hence, it suffices to show that for every $1\le k\le n/2$, 
	\[
	\frac{1}{(x+2)^2}
	> \frac{y^2}{(n-1)(xy+2)}\cdot\frac{k(n-k)}{kxy+n}.
	\]
	This inequality is equivalent to 
	\begin{equation}
	\begin{split}
		h(y) &\coloneqq a_{n,k}(x)y^2 - (n+2k)xy - 2n\\
		&= \frac{k(n-k)}{n-1}(x+2)^2y^2 - (xy+2)(kxy+n) < 0
		\label{eq08}
	\end{split}
	\end{equation}
	(see also \eqref{eq03} and \eqref{eq04}).
	If $a_{n,k}(x)\le0$, then inequality~\eqref{eq08} is trivial.
	If $a_{n,k}(x)>0$, then inequality~\eqref{eq08} holds 
	by the inequality $0<y<F_n(x)\le f_{n,k}(x)$ and fact~1 of Lemma~\ref{lem5}.
	Therefore, $M_n^\CQ(\epsilon;J_{1/2}) > M_n^\EC(\epsilon_1;J_{1/2})$.
	\par
	Let $\epsilon'\in\cE_n(\epsilon)$.
	Then $M_n^\EC(\epsilon';J_{1/2}) \ge M_n^\CQ(\epsilon;J_{1/2}) > M_n^\EC(\epsilon_1;J_{1/2})$.
	Since $\EC_n(\epsilon)$ is monotonically increasing in $\epsilon>0$, 
	so is $M_n^\EC(\epsilon;J_{1/2})$.
	Thus, $\epsilon'>\epsilon_1$, i.e., 
	\[
	\frac{e^{\epsilon'}-1}{e^\epsilon-1}
	> \frac{e^{\epsilon_1}-1}{e^\epsilon-1}.
	\]
	Since $\epsilon_1>0$ is arbitrary as long as $e^{\epsilon_1}-1<(e^\epsilon-1)F_n(e^\epsilon-1)$, 
	we obtain 
	\[
	\frac{e^{\epsilon'}-1}{e^\epsilon-1}
	\ge F_n(e^\epsilon-1).
	\]
\end{proof}

Theorem~\ref{main3} and Lemma~\ref{lem5} yield the following corollaries.

\begin{corollary}
	For all $n\ge3$, $\epsilon>0$ and $\epsilon'\in\cE_n(\epsilon)$, 
	\[
	\frac{e^{\epsilon'}-1}{e^\epsilon-1}
	> \sqrt{\frac{n(n-1)}{2\lfloor{n/2}\rfloor\lceil{n/2}\rceil}}.
	\]
\end{corollary}
\begin{proof}
	The assertion follows from Theorem~\ref{main3}, facts~4 and 5 of Lemma~\ref{lem5} immediately.
\end{proof}

\begin{corollary}
	For all $n\ge3$, $\epsilon>0$ with $e^\epsilon\ge2n-5$, and $\epsilon'\in\cE_n(\epsilon)$, 
	\begin{equation*}
		\frac{e^{\epsilon'}-1}{e^\epsilon-1}
		\ge f_{n,1}(e^\epsilon-1)
		> \frac{n+2}{4}(1-e^{-\epsilon}),
	\end{equation*}
	where $f_{n,1}$ is defined in Theorem~$\ref{main3}$.
\end{corollary}
\begin{proof}
	If $n=3$, the assertion follows from Theorem~\ref{main3} and $F_3(x)=f_{3,1}(x)$.
	If $n\ge4$, the assertion follows from Theorem~\ref{main3}, facts~7 and 8 of Lemma~\ref{lem5}.
\end{proof}

\section{Proof of Theorem~$\ref{main0}$}\label{proof0}

In this section, we prove Theorem~\ref{main0}.
First, let us begin with the classical optimization, 
for which we need the following definition and lemma \cite[Theorem~4]{Kairouz}.

\begin{definition}[Sublinear function]
	We say that a function $\phi\colon (0,\infty)^n\to\mathbb{R}$ is sublinear 
	if $\phi(x+y)\le\phi(x)+\phi(y)$ and $\phi(\alpha x)=\alpha\phi(x)$ for all $x,y\in(0,\infty)^n$ and $\alpha>0$.
\end{definition}

\begin{lemma}\label{lem3}
	Let $\Phi_\C$ be a real-valued function of $n$ probability vectors with the following condition: 
	there exists a sublinear function $\phi\colon (0,\infty)^n\to\mathbb{R}$ such that 
	\begin{equation}
		\Phi_\C(p_1,\ldots,p_n) = \sum_{p_1(k),\ldots,p_n(k)>0} \phi(p_1(k),\ldots,p_n(k)),
		\label{eq-sum-subl}
	\end{equation}
	where the above sum is taken all over $k$ with $p_1(k),\ldots,p_n(k)>0$.
	Then, for all $\epsilon>0$ and $n\ge2$, 
	\begin{align*}
		&\quad \sup_{(p_i)_{i=1}^n\in\C_n(\epsilon)} \Phi_\C(p_1,\ldots,p_n)\\
		&= \max\biggl\{ \sum_{v\in\cS_n(\epsilon)} \phi(v(1),\ldots,v(n))\alpha_v :
		\begin{array}{c}
			\sum_{v\in\cS_n(\epsilon)} \alpha_v v=\mathbf{1}_n,\\
			\forall v\in\cS_n(\epsilon),\ \alpha_v\ge0
		\end{array}
		\biggr\},
	\end{align*}
	where $\cS_n(\epsilon) \coloneqq \{1,e^\epsilon\}^n$.
\end{lemma}

Many information-theoretic quantities can be expressed as \eqref{eq-sum-subl}.
Such examples are relative entropy, Fisher information, total variation distance. 
Especially, $J_\theta(p,q)$ is expressed as 
\[
J_\theta(p,q) = \sum_{p(k),q(k)>0} f_\theta(p(k),q(k)),
\]
where the above sum is taken all over $k$ with $p(k),q(k)>0$, 
and the function $f_\theta$ defined in Theorem~\ref{main0} is sublinear.
We now prove the following lemma by using Lemma~\ref{lem3}.

\begin{lemma}\label{maxC}
	Let $\psi\colon (0,\infty)^2\to\mathbb{R}$ be a sublinear function with $\psi(1,1)=0$, 
	and $\Psi$ be the function $\Psi(p,q) = \sum_{p(k),q(k)>0} \psi(p(k),q(k))$ of two probability vectors.
	Then, for all $\epsilon>0$ and $n\ge2$, 
	\begin{align*}
		M_n^\C(\epsilon;\Psi) &\coloneqq \sup_{(p_i)_{i=1}^n\in\C_n(\epsilon)} \min_{i\not=j} \Psi(p_i,p_j)
		= \sup_{(p_i)_{i=1}^n\in\C_n(\epsilon)} \avg_{i\not=j} \Psi(p_i,p_j)\\
		&= \frac{\psi(e^\epsilon,1)+\psi(1,e^\epsilon)}{n-1}\max_{1\le k\le n/2} \frac{k(n-k)}{ke^\epsilon+n-k}.
	\end{align*}
\end{lemma}
\begin{proof}
	Let $\epsilon>0$ be a real number and $n\ge2$ be an integer.
	The following inequality holds: 
	\begin{equation}
		M_n^\C(\epsilon;\Psi)
		\le \sup_{(p_i)_{i=1}^n\in\C_n(\epsilon)} \avg_{i\not=j} \Psi(p_i,p_j).
		\label{eq09}
	\end{equation}
	Recall the definition $\cS_n(\epsilon)=\{1,e^\epsilon\}^n$.
	Set 
	\[
	\Phi_\C(p_1,\ldots,p_n) = \sum_{i\not=j} \Psi(p_i,p_j) \quad\text{and}\quad
	\phi(x) = \sum_{i\not=j} \psi(x(i),x(j))
	\]
	for $x\in(0,\infty)^n$.
	Lemma~\ref{lem3} yields that 
	\begin{equation}
	\begin{split}
		&\quad \sup_{(p_i)_{i=1}^n\in\C_n(\epsilon)} \sum_{i\not=j} \Psi(p_i,p_j)\\
		&= \max\biggl\{ \sum_{v\in\cS_n(\epsilon)} \sum_{i\not=j} \psi(v(i),v(j))\alpha_v :
		\begin{array}{c}
			\sum_{v\in\cS_n(\epsilon)} \alpha_v v=\mathbf{1}_n,\\
			\forall v\in\cS_n(\epsilon),\ \alpha_v\ge0
		\end{array}
		\biggr\}.
	\end{split}\label{eq10}
	\end{equation}
	Consider the partition of $\cS_n(\epsilon)$ into the $n+1$ subsets 
	\[
	\cS_{n,k}(\epsilon) \coloneqq \{ v\in\cS_n(\epsilon) : \text{the number of $i$ with $v(i)=e^\epsilon$ is }k \}
	\quad(k=0,1,\ldots,n).
	\]
	If $v\in\cS_{n,k}(\epsilon)$, 
	then $\sum_{i\not=j} \psi(v(i),v(j)) = \bigl( \psi(e^\epsilon,1)+\psi(1,e^\epsilon) \bigr)k(n-k)$ 
	due to the assumption $\psi(1,1)=0$.
	Thus, for every $k=0,1,\ldots,n$, we have 
	\begin{align*}
		\sum_{v\in\cS_n(\epsilon)} \sum_{i\not=j} \psi(v(i),v(j))\alpha_v
		&= \sum_{k=0}^n \sum_{v\in\cS_{n,k}(\epsilon)} \sum_{i\not=j} \psi(v(i),v(j))\alpha_v\\
		&= \sum_{k=0}^n \bigl( \psi(e^\epsilon,1)+\psi(1,e^\epsilon) \bigr)k(n-k)\beta_k,
	\end{align*}
	where $\beta_k=\sum_{v\in\cS_{n,k}(\epsilon)} \alpha_v$.
	Since the equality $\sum_{v\in\cS_n(\epsilon)} \alpha_v v=\mathbf{1}_n$ yields 
	\[
	\sum_{k=0}^n (ke^\epsilon+n-k)\beta_k
	= \sum_{v\in\cS_n(\epsilon)} \alpha_v\braket{\mathbf{1}_n|v}
	= \braket{\mathbf{1}_n|\mathbf{1}_n} = n,
	\]
	the right-hand side in \eqref{eq10} is bounded above by 
	\begin{align}
		&\quad \max\biggl\{ \sum_{k=0}^n \bigl( \psi(e^\epsilon,1)+\psi(1,e^\epsilon) \bigr)k(n-k)\beta_k :
		\begin{array}{c}
			\sum_{k=0}^n (ke^\epsilon+n-k)\beta_k=n,\\
			\beta_0,\ldots,\beta_n\ge0
		\end{array}
		\biggr\}\nonumber\\
		&= \bigl( \psi(e^\epsilon,1)+\psi(1,e^\epsilon) \bigr)n\max_{0\le k\le n} \frac{k(n-k)}{ke^\epsilon+n-k}\nonumber\\
		&= \bigl( \psi(e^\epsilon,1)+\psi(1,e^\epsilon) \bigr)n\max_{1\le k\le n/2} \frac{k(n-k)}{ke^\epsilon+n-k}. \label{eq11}
	\end{align}
	From \eqref{eq09}, \eqref{eq10} and \eqref{eq11}, it follows that 
	\begin{equation}
		M_n^\C(\epsilon;\Psi) \le \sup_{(p_i)_{i=1}^n\in\C_n(\epsilon)} \avg_{i\not=j} \Psi(p_i,p_j)
		\le \frac{\psi(e^\epsilon,1)+\psi(1,e^\epsilon)}{n-1}\max_{1\le k\le n/2} \frac{k(n-k)}{ke^\epsilon+n-k}.
		\label{eq12}
	\end{equation}
	\par
	Fix an arbitrary integer $1\le k\le n/2$.
	Let $d$ be the number of elements in $\cS_{n,k}(\epsilon)$, i.e., $d=\binom{n}{k}$.
	Then the vector space $\mathbb{R}^{\cS_{n,k}(\epsilon)}$ is isomorphic to $\mathbb{R}^d$.
	Define the probability vectors $p_1,\ldots,p_n\in\mathbb{R}^{\cS_{n,k}(\epsilon)}$ as 
	\[
	p_i(v) = \frac{v(i)}{\binom{n-1}{k-1}e^\epsilon+\binom{n-1}{k}}
	\quad(v\in\cS_{n,k}(\epsilon);\ i=1,\ldots,n).
	\]
	Then $(p_i)_{i=1}^n$ is $\epsilon$-DP, and moreover, 
	\begin{align*}
		&\quad \min_{i\not=j} \Psi(p_i,p_j)
		= \min_{i\not=j} \frac{1}{\binom{n-1}{k-1}e^\epsilon+\binom{n-1}{k}}\sum_{v\in\cS_{n,k}(\epsilon)} \psi(v(i),v(j))\\
		&= \frac{1}{\binom{n-1}{k-1}e^\epsilon+\binom{n-1}{k}}\cdot\bigl( \psi(e^\epsilon,1)+\psi(1,e^\epsilon) \bigr)\binom{n-2}{k-1}\\
		&= \frac{\psi(e^\epsilon,1)+\psi(1,e^\epsilon)}{((n-1)/(n-k))e^\epsilon+(n-1)/k}
		= \frac{\psi(e^\epsilon,1)+\psi(1,e^\epsilon)}{n-1}\cdot\frac{k(n-k)}{ke^\epsilon+n-k}.
	\end{align*}
	Since $1\le k\le n/2$ is arbitrary, 
	the inequalities in \eqref{eq12} turn to equality.
\end{proof}

Next, we consider the quantum optimization, 
for which we need the following lemmas.

\begin{lemma}\label{lem1}
	Let $n\ge2$ be an integer and $c\in[0,1]$ be a real number.
	There exists an $n$-tuple $(u_i)_{i=1}^n$ of unit vectors in $\mathbb{R}^n$ such that 
	$\braket{u_i|u_j}=c$ for all $i\not=j$.
\end{lemma}
\begin{proof}
	Since the matrix $A=(1-c)I_n+c\ketbra{\mathbf{1}_n}{\mathbf{1}_n}$ is positive semi-definite and consists of real numbers, 
	there exists a real square matrix $B$ of order $n$ such that $A=B^\top B$.
	The column vectors $u_1,\ldots,u_n\in\mathbb{R}^n$ of $B$ satisfy that 
	$\braket{u_i|u_j}=1$ if $i=j$ and $\braket{u_i|u_j}=c$ if $i\not=j$.
\end{proof}

\begin{lemma}\label{lem2}
	Let $c\in[0,1)$ and $\epsilon,t>0$ be real numbers, 
	$(u_i)_{i=1}^n$ be the $n$-tuple in Lemma~$\ref{lem1}$, 
	and $\rho_1,\ldots,\rho_n$ be the density matrices defined as 
	\[
	\rho_i = \frac{1}{n+t}(I_n + t\ketbra{u_i}{u_i}) \quad (i=1,\ldots,n).
	\]
	If $D=(e^\epsilon-1)^2+4(1-c^2)e^\epsilon$ and 
	\begin{equation*}
		0 < t \le t_{\max} \coloneqq \frac{2(e^\epsilon-1)}{\sqrt{D}+1-e^\epsilon},
	\end{equation*}
	then $(\rho_i)_{i=1}^n$ is CQ $\epsilon$-DP.
\end{lemma}
\begin{proof}
	Let $i\not=j$ be positive integers less than or equal to $n$. We show that 
	\begin{equation}
		\ketbra{u_i}{u_i} - e^\epsilon\ketbra{u_j}{u_j}
		\le \frac{1-e^\epsilon+\sqrt{D}}{2}I_n. \label{eq13}
	\end{equation}
	Take an orthonormal system $(e_i)_{i=1}^2$ of $\mathbb{C}^n$ such that 
	$u_i=e_1$, $u_j=\alpha e_1+\beta e_2$, $\alpha=c$ and $\beta=\sqrt{1-\alpha^2}$.
	Then the matrix $\ketbra{u_i}{u_i} - e^\epsilon\ketbra{u_j}{u_j}$ can be expressed as 
	a square matrix of order $2$: 
	\begin{equation}
		\begin{bmatrix}
			1&0\\
			0&0
		\end{bmatrix}
		- e^\epsilon
		\begin{bmatrix}
			\alpha^2&\alpha\beta\\
			\alpha\beta&\beta^2
		\end{bmatrix}
		=
		\begin{bmatrix}
			1-e^\epsilon\alpha^2&-e^\epsilon\alpha\beta\\
			-e^\epsilon\alpha\beta&-e^\epsilon\beta^2
		\end{bmatrix}
		\eqqcolon A. \label{eq14}
	\end{equation}
	Since $\Tr A=1-e^\epsilon$, $\det A=-e^\epsilon\beta^2$ and $D=(\Tr A)^2-4\det A$, 
	the greatest eigenvalue of $A$ is equal to $(1-e^\epsilon+\sqrt{D})/2$.
	Therefore, inequality~\eqref{eq13} holds.
	Consequently, we obtain 
	\[
	t(\ketbra{u_i}{u_i} - e^\epsilon\ketbra{u_j}{u_j})
	\le t_{\max}\frac{1-e^\epsilon+\sqrt{D}}{2}I_n
	= (e^\epsilon-1)I_n
	\]
	for all $0<t\le t_{\max}$.
	This implies $\rho_i\le e^\epsilon\rho_j$.
\end{proof}

Recalling the definition of $M_n^{\mathrm{X}}(\epsilon;J_\theta)$, 
we have the monotonicity 
\begin{equation}
	M_2^{\mathrm{X}}(\epsilon;J_\theta) \ge M_3^{\mathrm{X}}(\epsilon;J_\theta) \ge \cdots
	\label{eq15}
\end{equation}
for $\mathrm{X}=\C,\EC,\CQ$.
This monotonicity is used below.

\begin{lemma}\label{maxCQ}
	For all $\theta\in[0,1]$, $\epsilon>0$ and $n\ge2$, 
	$M_n^\CQ(\epsilon;J_\theta) = M_2^\C(\epsilon;J_\theta)$.
\end{lemma}
\begin{proof}
	Let $\theta\in[0,1]$, $c\in[0,1)$ and $\epsilon,t>0$ be real numbers, 
	and $n\ge2$ be an integer.
	Take a CQ $\epsilon$-DP $n$-tuple $(\rho_i)_{i=1}^n$ in Lemma~\ref{lem2}, i.e., 
	\[
	\rho_i = \frac{1}{n+t}(I_n + t\ketbra{u_i}{u_i}),\quad
	u_i\in\mathbb{R}^n \quad (i=1,\ldots,n),
	\]
	$\braket{u_i|u_j}=1$ if $i=j$ and $\braket{u_i|u_j}=c$ if $i\not=j$.
	Fix $i\not=j$ arbitrarily.
	The matrix $\ketbra{u_i}{u_i}-\ketbra{u_j}{u_j}$ can be expressed as 
	a square matrix of order $2$ in the same way as \eqref{eq14}: 
	\[
	\begin{bmatrix}
		1&0\\
		0&0
	\end{bmatrix}
	-
	\begin{bmatrix}
		\alpha^2&\alpha\beta\\
		\alpha\beta&\beta^2
	\end{bmatrix}
	=
	\begin{bmatrix}
		\beta^2&-\alpha\beta\\
		-\alpha\beta&-\beta^2
	\end{bmatrix}
	= -\beta
	\begin{bmatrix}
		-\beta&\alpha\\
		\alpha&\beta
	\end{bmatrix},
	\]
	where $\alpha=c$ and $\beta=\sqrt{1-\alpha^2}$.
	Moreover, $(\ketbra{u_i}{u_i}-\ketbra{u_j}{u_j})^2$ is expressed as 
	\[
	\beta^2
	\begin{bmatrix}
		-\beta&\alpha\\
		\alpha&\beta
	\end{bmatrix}^2
	= \beta^2I_2.
	\]
	Denote by $\lambda_1$ and $\lambda_2$ two eigenvalues of the matrix 
	\begin{equation*}
		A \coloneqq I_2+t\Bigl(
		(1-\theta)
		\begin{bmatrix}
			1&0\\
			0&0
		\end{bmatrix}
		+ \theta
		\begin{bmatrix}
			\alpha^2&\alpha\beta\\
			\alpha\beta&\beta^2
		\end{bmatrix}
		\Bigr)
		=
		\begin{bmatrix}
			1+t(1-\theta\beta^2)&t\theta\alpha\beta\\
			t\theta\alpha\beta&1+t\theta\beta^2
		\end{bmatrix}.
	\end{equation*}
	It follows that 
	\begin{align*}
		&\quad J_\theta(\rho_i,\rho_j)
		= \frac{t^2}{n+t}\Tr(\ketbra{u_i}{u_i}-\ketbra{u_j}{u_j})^2\Bigl( I_n+t\bigl( (1-\theta)\ketbra{u_i}{u_i}+\theta\ketbra{u_j}{u_j} \bigr) \Bigr)^{-1}\\
		&= \frac{t^2}{n+t}\beta^2\Tr A^{-1}
		= \frac{t^2}{n+t}\beta^2(1/\lambda_1+1/\lambda_2)
		= \frac{t^2}{n+t}\beta^2\frac{\lambda_1 + \lambda_2}{\lambda_1\lambda_2}
		= \frac{t^2}{n+t}\beta^2\frac{\Tr A}{\det A}.
	\end{align*}
	Since $\Tr A=2+t$ and 
	\begin{align*}
		\det A &= \bigl( 1+t(1-\theta\beta^2) \bigr)(1+t\theta\beta^2) - (t\theta\alpha\beta)^2\\
		&= 1 + t + t^2(1-\theta\beta^2)\theta\beta^2 - (t\theta\beta)^2\alpha^2\\
		&= 1 + t + t^2\theta\beta^2 - (t\theta\beta)^2
		= 1 + t + t^2\theta(1-\theta)\beta^2\\
		&= 1 + t + t^2\theta(1-\theta)(1-c^2),
	\end{align*}
	we have 
	\begin{equation}
		J_\theta(\rho_i,\rho_j)
		= \frac{t^2}{n+t}\beta^2\frac{\Tr A}{\det A}
		= \frac{t^2}{n+t}\cdot(1-c^2)\cdot\frac{2+t}{1 + t + t^2\theta(1-\theta)(1-c^2)}.
		\label{eq16}
	\end{equation}
	\par
	Finally, putting $t=t_{\max}$, we take the limit $c\to1-0$.
	Set $s=e^\epsilon-1>0$.
	Then $D=s^2+4(1-c^2)e^\epsilon$ and 
	\[
	t_{\max} = \frac{2s}{\sqrt{D}-s}
	= \frac{2s(\sqrt{D}+s)}{D-s^2}
	= \frac{s(\sqrt{D}+s)}{2(1-c^2)e^\epsilon}.
	\]
	Thus, the positive number $t=t_{\max}$ diverges to $+\infty$ as $c\to1-0$, 
	and moreover, 
	\begin{gather*}
		\lim_{c\to1-0} \frac{t}{n+t} = 1,\quad
		\lim_{c\to1-0} (1-c^2)t = \frac{s^2}{e^\epsilon},\\
		\begin{split}
			&\quad \lim_{c\to1-0} \frac{2+t}{1 + t + t^2\theta(1-\theta)(1-c^2)}
			= \lim_{c\to1-0} \frac{t}{t + t^2\theta(1-\theta)(1-c^2)}\\
			&= \lim_{c\to1-0} \frac{1}{1 + t\theta(1-\theta)(1-c^2)}
			= \frac{1}{1+\theta(1-\theta)s^2/e^\epsilon}
			= \frac{e^\epsilon}{e^\epsilon+\theta(1-\theta)s^2}.
		\end{split}
	\end{gather*}
	Since Proposition~\ref{n=2} and Lemma~\ref{maxC} yield that 
	\begin{align*}
		&\quad M_2^\CQ(\epsilon;J_\theta)
		= M_2^\C(\epsilon;J_\theta)
		= \frac{f_\theta(e^\epsilon,1)+f_\theta(1,e^\epsilon)}{e^\epsilon+1}\\
		&= \frac{1}{e^\epsilon+1}\cdot\frac{(e^\epsilon-1)^2(e^\epsilon+1)}{((1-\theta)e^\epsilon+\theta)(\theta e^\epsilon+1-\theta)}
		= \frac{s^2}{((1-\theta)s+1)(\theta s+1)},
	\end{align*}
	it turns out that 
	\begin{align*}
		&\quad M_2^\C(\epsilon;J_\theta) = M_2^\CQ(\epsilon;J_\theta) \ge M_n^\CQ(\epsilon;J_\theta)
		\ge \lim_{c\to1-0} J_\theta(\rho_1,\rho_2)\\
		&= \lim_{c\to1-0} \frac{t^2}{n+t}\cdot(1-c^2)\cdot\frac{2+t}{1 + t + t^2\theta(1-\theta)(1-c^2)}\\
		&= \frac{s^2}{e^\epsilon}\cdot\frac{e^\epsilon}{e^\epsilon+\theta(1-\theta)s^2}
		= \frac{s^2}{e^\epsilon+\theta(1-\theta)s^2}
		= \frac{s^2}{((1-\theta)s+1)(\theta s+1)}\\
		&= M_2^\C(\epsilon;J_\theta),
	\end{align*}
	where inequality~\eqref{eq15} has been used to obtain the first inequality.
\end{proof}

\begin{proof}[Proof of Theorem~$\ref{main0}$]
	The former assertion follows from Lemmas~\ref{maxC} and \ref{maxCQ} immediately.
	Let us show the latter assertion.
	Let $\theta\in[0,1]$ and $\epsilon>0$ be real numbers, 
	and $n\ge3$ be an integer.
	Since Lemmas~\ref{maxC}, \ref{maxCQ} and inequality~\eqref{eq15} imply that 
	\begin{align*}
		M_n^\C(\epsilon;J_\theta) &\le M_3^\C(\epsilon;J_\theta)
		= \frac{f_\theta(e^\epsilon,1)+f_\theta(1,e^\epsilon)}{2}\cdot\frac{2}{e^\epsilon+2}\\
		&= \frac{e^\epsilon+1}{e^\epsilon+2}M_2^\C(\epsilon;J_\theta)
		< M_2^\C(\epsilon;J_\theta) = M_n^\CQ(\epsilon;J_\theta),
	\end{align*}
	we obtain the latter assertion.
\end{proof}

\section{Concrete CQ $\epsilon$-DP $n$-tuples that do not lie in $\EC_n(\epsilon)$}\label{concrete}

Using Corollary~\ref{main4}, we construct CQ $\epsilon$-DP $n$-tuples that do not lie in $\EC_n(\epsilon)$.
In this section, we use the following lemmas instead of Lemmas~\ref{lem1} and \ref{lem2}.

\begin{lemma}\label{lem1'}
	Let $d\ge2$ be an integer and $c\in[1/d,1]$ be a real number.
	There exists a $(d+1)$-tuple $(u_i)_{i=1}^{d+1}$ of unit vectors in $\mathbb{C}^d$ such that 
	$\abs{\braket{u_i|u_j}}=c$ for all $i\not=j$.
\end{lemma}
\begin{proof}
	For $z\in\mathbb{C}$, 
	define the Hermitian matrix $A(z)=(\alpha_{i,j})$ of order $d+1$ as 
	$\alpha_{i,j}=1$ if $i=j$ and $\alpha_{i,j}=z$ if $i<j$.
	Denote by $\eig A(z)$ the set of all eigenvalues of $A(z)$.
	Then $\eig A(c)=\{1+dc,1-c\}$, $\eig A(-c)=\{1-dc,1+c\}$, 
	$\min\eig A(c)=1-c\ge0$ and $\min\eig A(-c)=1-dc\le0$.
	Since the minimum eigenvalue of $A(z)$ can be expressed as 
	\[
	\min\eig A(z) = \min_{\|u\|=1} \braket{u |A(z)| u},
	\]
	it follows that 
	\[
	\abs{\min\eig A(z) - \min\eig A(z')} \le \|A(z)-A(z')\|
	\]
	for all $z,z'\in\mathbb{C}$, where $\|\cdot\|$ denotes the operator norm.
	This shows that $\min\eig A(z)$ is continuous in $z$.
	Thus, the intermediate value theorem implies that 
	$\min\eig A(z_0)=0$ for some $z_0\in\mathbb{C}$ of magnitude $c$.
	Therefore, there exists a $d\times(d+1)$ complex matrix $B$ such that $A(z_0)=B^\ast B$.
	The column vectors $u_1,\ldots,u_{d+1}\in\mathbb{C}^d$ of $B$ satisfy that 
	$\braket{u_i|u_j}=1$ if $i=j$ and $\braket{u_i|u_j}=z_0$ if $i<j$.
\end{proof}

\begin{lemma}\label{lem2'}
	Let $c\in[1/d,1)$ and $\epsilon,t>0$ be real numbers, 
	$(u_i)_{i=1}^{d+1}$ be the $(d+1)$-tuple in Lemma~$\ref{lem1'}$, 
	and $\rho_1,\ldots,\rho_{d+1}$ be the density matrices defined as 
	\[
	\rho_i = \frac{1}{d+t}(I_d + t\ketbra{u_i}{u_i}) \quad (i=1,\ldots,d+1).
	\]
	If $D=(e^\epsilon-1)^2+4(1-c^2)e^\epsilon$ and 
	\begin{equation*}
		0 < t \le t_{\max} \coloneqq \frac{2(e^\epsilon-1)}{\sqrt{D}+1-e^\epsilon},
	\end{equation*}
	then $(\rho_i)_{i=1}^{d+1}$ is CQ $\epsilon$-DP.
\end{lemma}
\begin{proof}
	See the proof of Lemma~\ref{lem2}.
\end{proof}

\begin{theorem}\label{thm1}
	Let $(\rho_i)_{i=1}^3$ be a CQ $\epsilon$-DP $3$-tuple in Lemma~$\ref{lem2'}$ with $d=2$ and $t=t_{\max}$.
	Then $(\rho_i)_{i=1}^3$ does not lie in $\EC_3(\epsilon)$.
\end{theorem}
\begin{proof}
	Set $s=e^\epsilon-1>0$.
	Then $D=s^2+4(1-c^2)(s+1)$ and $t=t_{\max}=2s/(\sqrt{D}-s)$.
	Let $i\not=j$ be positive integers less than or equal to $3$.
	We show that $M_3^\C(\epsilon;J_{1/2}) < J_{1/2}(\rho_i,\rho_j)$.
	First, Theorem~\ref{main0} implies that 
	\[
	M_3^\C(\epsilon;J_{1/2})
	= \frac{4(e^\epsilon-1)^2}{2(e^\epsilon+1)}\cdot\frac{2}{e^\epsilon+2}
	= \frac{4s^2}{(s+2)(s+3)}.
	\]
	Also, it follows form the same calculation as \eqref{eq16} that 
	\[
	J_{1/2}(\rho_i,\rho_j)
	= \frac{t^2}{2+t}\cdot(1-c^2)\cdot\frac{2+t}{1 + t + (t/2)^2(1-c^2)}
	= \frac{(1-c^2)t^2}{1 + t + (t/2)^2(1-c^2)},
	\]
	where we must replace $n$ in \eqref{eq16} with the dimension $d=2$.
	Thus, 
	\begin{align*}
		&\quad M_3^\C(\epsilon;J_{1/2}) < J_{1/2}(\rho_i,\rho_j)
		\iff \frac{4s^2}{(s+2)(s+3)} < \frac{(1-c^2)t^2}{1 + t + (t/2)^2(1-c^2)}\\
		&\iff \frac{(s+2)(s+3)}{s^2} > \frac{4}{1-c^2}\cdot\frac{1 + t + (t/2)^2(1-c^2)}{t^2}\\
		&\iff 1+\frac{5s+6}{s^2} > \frac{4}{1-c^2}\Bigl( \frac{1+t}{t^2} + \frac{1-c^2}{4} \Bigr)
		\iff \frac{5s+6}{s^2} > \frac{4}{1-c^2}\cdot\frac{1+t}{t^2}.
	\end{align*}
	Recalling that $D=s^2+4(1-c^2)(s+1)$ and $t=t_{\max}=2s/(\sqrt{D}-s)$, we have 
	\begin{align*}
		\frac{t^2}{1+t} &= t-1+\frac{1}{1+t}
		= \frac{2s}{\sqrt{D}-s}-1+\frac{\sqrt{D}-s}{\sqrt{D}+s}
		= \frac{2s}{\sqrt{D}-s}+\frac{-2s}{\sqrt{D}+s}\\
		&= \frac{4s^2}{D-s^2}
		= \frac{s^2}{(1-c^2)(s+1)}.
	\end{align*}
	Therefore, 
	\begin{equation*}
		M_3^\C(\epsilon;J_{1/2}) < J_{1/2}(\rho_i,\rho_j)
		\iff \frac{5s+6}{s^2} > \frac{4(s+1)}{s^2}.
	\end{equation*}
	Since the right inequality always holds, 
	so does the left inequality.
	By Corollary~\ref{main4}, 
	$(\rho_i)_{i=1}^3$ does not lie in $\EC_3(\epsilon)$.
\end{proof}

\begin{corollary}
	Let $n\ge3$ be an integer, and 
	$(\rho_i)_{i=1}^3$ be a CQ $\epsilon$-DP $3$-tuple in Lemma~$\ref{lem2'}$ with $d=2$ and $t=t_{\max}$.
	Then every $(\sigma_i)_{i=1}^n\in\CQ_n(\epsilon)$ with $\sigma_i=\rho_i$, $i=1,2,3$, does not lie in $\EC_n(\epsilon)$.
\end{corollary}
\begin{proof}
	If $(\sigma_i)_{i=1}^n\in\EC_n(\epsilon)$, then $(\sigma_i)_{i=1}^3\in\EC_3(\epsilon)$.
	Thus, the assertion follows from Theorem~\ref{thm1} and Definition~\ref{EC}.
\end{proof}

\section{Conclusion}

We have investigated the difference between the sets $\EC_n(\epsilon)$ and $\CQ_n(\epsilon)$, 
which is represented by the infimum $\epsilon_{\inf}=\epsilon_{\inf}(n,\epsilon)=\inf\cE_n(\epsilon)$.
This infimum has the upper and lower bounds as \eqref{eq01}.
The lower bound can probably be improved, 
but we do not know whether the upper bound can be improved or not.
It is desirable to find tighter bounds for $\epsilon_{\inf}$.

Although we have not fixed the dimension $d$ of the finite-dimensional vector spaces $\mathbb{C}^d$ and $\mathbb{R}^d$, 
it is also important to study the case when $d$ is fixed.
For instance, it is an interesting problem to find extreme points of $\CQ_n^{(d)}(\epsilon)$.
For the classical case, 
Holohan et al.\ \cite{Holohan2} studied extreme points of $\C_n^{(d)}(\epsilon)$.

We have used Lemma~\ref{lem1'} to construct CQ $\epsilon$-DP $n$-tuples that do not lie in $\EC_n(\epsilon)$.
Instead of Lemma~\ref{lem1'}, 
one might use \textit{symmetric, informationally complete, positive-operator-valued measures (SIC-POVMs)}.
In this case, one can probably prove that 
the CQ $\epsilon$-DP $d^2$-tuple $(\rho_i)_{i=1}^{d^2}$ of density matrices on $\mathbb{C}^d$ 
constructed by a SIC-POVM does not lie in $\EC_{d^2}(\epsilon)$.
However, we can prove this statement only for large $\epsilon>0$ if using Corollary~\ref{main4}.
Hence, one needs an alternative criterion instead of Corollary~\ref{main4} to prove the above statement.

\section*{Acknowledgments}
Partial contents of this paper are also used in the author's doctoral thesis.
The author is grateful to Prof.\ Fran\c{c}ois Le Gall and Prof.\ Yoshimichi Ueda for giving me some advice on a part of this paper (precisely, the doctoral thesis).
The author was supported by JSPS KAKENHI Grant Number JP19J20161.

\section*{Appendix}
In this appendix, we discuss linear mappings from $\Herm(d)$ into $\Herm(d')$ 
which are used in quantum information theory.
Recall that $\Herm(d)$ is the set of all Hermitian matrices on $\mathbb{C}^d$.
Denote by $\PSD(d)$ the set of all positive semi-definite matrices on $\mathbb{C}^d$.
First, let us begin with several basic terms (see also a textbook in quantum information theory, e.g., \cite{book1,book2}).
For two linear mapping $\Lambda_i$, $i=1,2$, from $\Herm(d_i)$ into $\Herm(d'_i)$, 
the tensor product $\Lambda_1\otimes\Lambda_2$ is a linear mapping from $\Herm(d_1)\otimes\Herm(d_2)$ into $\Herm(d'_1)\otimes\Herm(d'_2)$.
Since $\Herm(d_1)\otimes\Herm(d_2)$ can be regarded as $\Herm(d_1d_2)$, 
the tensor product $\Lambda_1\otimes\Lambda_2$ is also a linear mapping from $\Herm(d_1d_2)$ into $\Herm(d'_1d'_2)$.
Let $\id_d$ be the identity mapping on $\Herm(d)$.
\begin{itemize}
	\item
	A linear mapping $\Lambda$ from $\Herm(d)$ into $\Herm(d')$ is called \textit{positive} 
	if $\Lambda(\PSD(d))\subset\PSD(d')$.
	\item
	A linear mapping $\Lambda$ from $\Herm(d)$ into $\Herm(d')$ is called \textit{completely positive} 
	if $\Lambda\otimes\id_k$ is positive for every integer $k\ge2$.
	\item
	A linear mapping $\Lambda$ from $\Herm(d)$ into $\Herm(d')$ is called \textit{trace-preserving} 
	if $\Tr\Lambda(X)=\Tr X$ for every $X\in\Herm(d)$.
	\item
	A linear mapping $\Lambda$ from $\Herm(d)$ into $\Herm(d')$ is called \textit{CPTP} 
	if $\Lambda$ is completely positive and trace-preserving.
\end{itemize}
In quantum information theory, a quantum channel is a CPTP map.

\begin{example*}[Entanglement breaking channel]
	The linear mapping $\Lambda$ below is a CPTP map called \textit{entanglement breaking channel}.
	Let $\sigma_1,\ldots,\sigma_m$ be density matrices on $\mathbb{C}^{d'}$ 
	and $(M_k)_{k=1}^m$ be a POVM, i.e., $M_1,\ldots,M_m\ge0$ and $\sum_{k=1}^m M_k=I_d$.
	For example, $(\ketbra{e_k}{e_k})_{k=1}^d$ is a POVM, 
	where $(e_k)_{k=1}^d$ denotes the standard basis of $\mathbb{C}^d$.
	Define the linear mapping $\Lambda$ from $\Herm(d)$ into $\Herm(d')$ as $\Lambda(X) = \sum_{k=1}^m (\Tr M_kX)\sigma_k$.
	It can easily be checked that $\Lambda$ is a CPTP map.
	This fact is used implicitly in this section.
\end{example*}

Next, we prove Proposition~$\ref{CPTP}$.

\begin{proof}[Proof of Proposition~$\ref{CPTP}$]
	First, assume that $(\rho_i)_{i=1}^n$ lies in $\EC_n(\epsilon)$, i.e., 
	there exist $(p_i)_{i=1}^n\in\C_n(\epsilon)$ and a CPTP map $\Lambda$ such that 
	$\rho_i=\Lambda(\diag(p_i))$ for all $i=1,\ldots,n$.
	Denote by $(e_k)_{k=1}^d$ the standard basis of $\mathbb{C}^d$, 
	where $d$ is the dimension of the vector space that $p_1,\ldots,p_n$ inhabit.
	Then $\rho_i = \sum_{k=1}^d p_i(k)\Lambda(\ketbra{e_k}{e_k})$ for all $i=1,\ldots,n$.
	Since all $\Lambda(\ketbra{e_k}{e_k})$ are density matrices, 
	$(\rho_i)_{i=1}^n$ lies in the right-hand side of \eqref{eq-CPTP}.
	\par
	Conversely, assume that $(\rho_i)_{i=1}^n$ lies in the right-hand side of \eqref{eq-CPTP}: 
	there exist $(p_i)_{i=1}^n\in\C_n(\epsilon)$ and density matrices $\sigma_k$ 
	such that $\rho_i=\sum_{k=1}^d p_i(k)\sigma_k$ for all $i=1,\ldots,n$, 
	where $d$ is the dimension of the vector space that $p_1,\ldots,p_n$ inhabit.
	Define the CPTP map $\Lambda$ as 
	$\Lambda(X) = \sum_{k=1}^d \braket{e_k |X| e_k}\sigma_k$.
	Then $\Lambda(\diag(p_i))=\rho_i$ for all $i=1,\ldots,n$.
	Therefore, $(\rho_i)_{i=1}^n$ lies in $\EC_n(\epsilon)$.
\end{proof}

\end{document}